\documentclass[11pt]{article}
\usepackage{natbib}
\usepackage{soul}
\usepackage{url}
\usepackage[utf8]{inputenc}
\usepackage[small]{caption}
\usepackage{graphicx}
\usepackage{amsmath}
\usepackage{amsthm}
\usepackage{amsfonts}
\usepackage{booktabs}
\usepackage{algorithm}
\usepackage{algorithmic}
\usepackage{fullpage}
\newtheorem{theorem}{Theorem}
\newtheorem{lemma}[theorem]{Lemma}
\newtheorem{claim}[theorem]{Claim}
\theoremstyle{definition}

\newcommand{\bigO}{\mathcal{O}}
\newcommand{\discretization}{\delta}
\newcommand{\perm}{\mathcal{L}(n)}

\DeclareMathOperator*{\SW}{\text{SW}}
\DeclareMathOperator*{\LP}{\text{LP}}
\DeclareMathOperator{\SUB}{{Sub}}
\DeclareMathOperator*{\EG}{\text{EG}}
\DeclareMathOperator{\sumv}{sum}
\newcommand{\wgt}{\text{wgt}}
\newcommand{\R}{\mathbb{R}}
\newcommand{\T}{\mathbf{T}}

\newcommand{\problem}{SMEF}

\allowdisplaybreaks

\title{\bf Computing envy-freeable allocations\\with limited subsidies}

\author{Ioannis Caragiannis\thanks{Department of Computer Science, Aarhus University, Aabogade 34, 8200 Aarhus N, Denmark. Email: \texttt{iannis@cs.au.dk}}
\and
Stavros Ioannidis\thanks{Department of Informatics, King's College London, Bush House, Strand Campus, 30 Aldwych, London WC2B 4BG, United Kingdom. Email: \texttt{stavros.ioannidis@kcl.ac.uk}}
}

\date{}

\begin{document}
	
\maketitle

\begin{abstract}
Fair division has emerged as a very hot topic in multiagent systems, and envy-freeness is among the most compelling fairness concepts. An allocation of indivisible items to agents is envy-free if no agent prefers the bundle of any other agent to his own in terms of value. As envy-freeness is rarely a feasible goal, there is a recent focus on relaxations of its definition. An approach in this direction is to complement allocations with payments (or subsidies) to the agents. A feasible goal then is to achieve envy-freeness in terms of the total value an agent gets from the allocation and the subsidies.
  	
We consider the natural optimization problem of computing allocations that are {\em envy-freeable} using the minimum amount of subsidies. As the problem is NP-hard, we focus on the design of approximation algorithms. On the positive side, we present an algorithm which, for a constant number of agents, approximates the minimum amount of subsidies within any required accuracy, at the expense of a graceful increase in the running time. On the negative side, we show that, for a superconstant number of agents, the problem of minimizing subsidies for envy-freeness is not only hard to compute exactly (as a folklore argument shows) but also, more importantly, hard to approximate. 
\end{abstract}

\section{Introduction}
Fairly dividing goods among people is an extremely important quest since antiquity. Today, fair division is a flourishing area of research in computer science, economics, and political science and {\em envy-freeness} is considered as the ultimate fairness concept~\citep{P20}. Following a research trend that is very popular in AI recently, we consider allocation problems with indivisible items. An allocation of items to agents is envy-free if no agent prefers the bundle of items allocated to some other agent to her own. Traditionally, agents' preferences are based on cardinal valuations they have for the items.

Unfortunately, with indivisible items, envy-freeness is rarely a feasible goal. For example, no such allocation exists in the embarrassingly simple case with a single item and two agents with some value for it. Recently proposed relaxations of envy-freeness aim to serve as useful alternative fairness notions. In a line of research that emerged very recently, allocations are complemented with payments (or {\em subsidies}) to the agents \citep{HS19,BDN+19}. Now, envy-freeness dictates that no agent prefers the allocation and payment of another agent to hers, and becomes a feasible goal. However, important questions arise related to the sparing use of money.

In this paper, we follow an optimization approach. We define and study the optimization problem \problem\ (standing for Subsidy Minimization for Envy-Freeness). Given an allocation problem consisting of items and agents with valuations for the items, \problem\ asks for an allocation that is envy-freeable using the minimum total amount of subsidies.

\problem\ is NP-hard; this follows by the NP-hardness of deciding whether a given allocation problem has an envy-free allocation or not. Thus, we resort to {\em approximation algorithms} for \problem. As multiplicative approximation guarantees are hopeless, our aim is to design algorithms that run in polynomial-time and compute an allocation that is envy-freeable with an amount of subsidies that does not exceed the minimum possible amount of subsidies (denoted $\chi$) by much. In particular, we use the total valuation of all agents for all goods (denoted by $\sumv{v}$) as a benchmark and seek allocations that are envy-freeable with an amount of at most $\chi+\rho\cdot \sumv{v}$ as subsidies. The goal for the approximation guarantee $\rho$ of an algorithm is to be as small as possible.

We initiate the study of \problem\ and present two results. On the positive side, we design an algorithm that achieves an arbitrary low approximation guarantee of $\epsilon>0$. When applied to allocation instances with a constant number of agents, the algorithm uses {\em dynamic programming} and runs in time that is polynomial in the number of items and $1/\epsilon$. On the negative side, we show that, in general, \problem\ is not only hard to solve exactly, but also hard to approximate within a small constant. Unlike the folklore reduction\footnote{Notice that deciding whether an envy-free allocation exists for two agents with identical item valuations requires solving {\sc Partition}, a well-known NP-hard problem~\citep{GJ79}.} for proving hardness of envy-freeness, our proof uses a novel {\em approximation-preserving} reduction. Besides separating the general case from that with constantly many agents, our negative result  indicates that achieving good approximation guarantees will be a challenging goal. 

\subsection{Related work} 
The concept of envy-freeness was formally introduced by~\cite{F67} and~\cite{V74}. As envy-freeness may not be achievable when goods are indivisible, recent research has focused on defining approximations of envy-freeness. These include envy-freeness up to one good \citep{B11}, envy-freeness up to any good~\citep{CKM+19}, epistemic envy-freeness~\citep{ABC+18}, and more. Still, achieving even them in polynomial time can be challenging, and recent work has focused on approximation algorithms; see, e.g., \citep{LMMS04,CEEM07,PR18,BKV18,CGH19,CKMS20,ANM20}.

The approach of mixing allocations with payments either from or to the agents has been extensively considered in the economics literature. A typical example is the rent division problem, where $n$ items (rooms) and a fixed rent have to be divided among $n$ agents in an envy-free manner \citep{S99,S83}. Compensations to the agents were first considered by~\cite{M87}. Subsequent papers consider unit-demand allocation problems, where each agent can get at most one item; see, e.g., \citep{ADG91}. \cite{A95} and \cite{K00} give polynomial-time algorithms that compute allocations and payments. More general models are studied by~\cite{HRS02} and~\cite{MPR02}.

In the AI literature, \cite{CEM17} consider allocation problems and monetary transfers between the agents. In a model that is the closest to ours, \cite{HS19} aim to bound the amount of external subsidies assuming that all agent valuations for goods are in $[0,1]$. Among several results, they conjectured that subsidies of $n-1$ suffice; an even stronger version of the conjecture was proved very recently by~\cite{BDN+19}.

\subsection{Roadmap}
The rest of the paper is structured as follows. We begin with preliminary definitions in Section~\ref{sec:prelim}. Our approximation algorithm is presented in Section~\ref{sec:alg} and our result on the hardness of approximation for SMEF is presented in Section~\ref{sec:hardness}. We conclude in Section~\ref{sec:open}.

\section{Preliminaries}\label{sec:prelim}
We consider allocation instances with a set $M$ of $m$ items and a set $N$ of $n$ agents. Each agent $i\in N$ has a valuation function $v_i:M\rightarrow \R_{\geq 0}$ over the items.\footnote{In our exposition, we assume that valuations are non-negative, even though our positive result can be extended to work without this assumption, in the model of~\cite{ACIW19} where items can be goods or chores.} With some abuse of notation, we use $v_i(B)$ to denote the valuation of agent $i$ for the set (or {\em bundle}) of items $B$. Valuations are additive, i.e., $v_i(B)=\sum_{g\in B}{v_i(g)}$. An allocation is simply a partition $X=(X_1, X_2, ..., X_n)$ of the items of $M$ into $n$ disjoint bundles, where agent $i\in N$ is supposed to get the bundle $X_i$. We use the abbreviations $\sumv{v}=\sum_{i\in N}{v_i(M)}$ and $\max{v}=\max_{i\in N}{v_i(M)}$. 

As usual, we define the {\em social welfare} of an allocation $X=(X_1, ..., X_n)$ to be $\SW(X,v)=\sum_{i\in N}{v_i(X_i)}$. An allocation $X=(X_1, X_2, ..., X_n)$ is {\em envy-free} if $v_i(X_i)\geq v_i(X_j)$ for every pair of agents $i$ and $j$. Informally, envy-freeness requires that no agent envies the bundle allocated to any other agent compared to her own.

For an allocation $X=(X_1, ..., X_n)$  in an instance with agent valuations $v$, the {\em envy graph}  $\EG(X,v)$, introduced by \cite{LMMS04}, is an edge-weighted complete directed graph that has a node for each agent and the weight of the directed edge $(i,j)$ represents the ``envy'' of agent $i$ for agent $j$. Using $G=\EG(X,v)$ and $\wgt_G(i,j)$ for the weight of the directed edge from node $i$ to node $j$ in the envy graph $\EG(X,v)$, 
we define $\wgt_G(i,j)=v_i(X_j)-v_i(X_i)$. 

Following the modelling assumptions of \cite{HS19}, we also consider {\em payments} (or {\em subsidies}) to the agents, represented by a payment vector $\pi=\langle \pi_1, ..., \pi_n\rangle$ with non-negative entries, i.e., $\pi_i\geq 0$ for every agent $i\in N$. Below, we use the terms ``payment'' and ``subsidy'' interchangeably. Now, we say that the pair $(X,\pi)$ of the allocation $X$ and payment vector $\pi$ is envy-free if $v_i(X_i)+\pi_i \geq v_i(X_j)+\pi_j$ for every pair of agents $i,j\in N$. Informally, this extended version of envy-freeness requires that no agent envies the bundle and the payment of any other agent compared to the bundle and payment she gets.

We say that allocation $X$ is {\em envy-freeable} if there is a payment vector $\pi$ so that the pair $(X,\pi)$ is envy-free.  Although the use of payments makes envy-freeness a feasible goal, not all allocations are envy-freeable. The following theorem, due to~\cite{HS19}, gives sufficient and necessary conditions so that an allocation is envy-freeable.

\begin{theorem}[\cite{HS19}]\label{thm:hs19}
The following statements are equivalent:
\begin{itemize}
	\item The allocation $X=(X_1,X_2, ..., X_n)$ is envy-freeable.
	\item The allocation $X$ maximizes social welfare among all redistributions of its bundles to the agents.
	\item The envy graph $\EG(X,v)$ contains no directed cycles of positive total weight.
\end{itemize}	
\end{theorem}

Detecting whether a given allocation $X$ is envy-freeable can be done using the following linear program $\LP(X,v)$:
\begin{align}\label{eq:LP}
	\mbox{minimize} & \quad \sum_{i\in N}{\pi_i}\\\nonumber
	\mbox{subject to:} & \quad \pi_i-\pi_j \geq v_i(X_j)-v_i(X_i), \forall i,j\in N\\\nonumber
	 & \quad \pi\geq 0
\end{align}
$\LP(X,v)$ aims to find a payment vector $\pi$ so that the envy-freeness constraints between pairs of agents are satisfied. In addition, it minimizes the total amount of payments. As it is observed by~\cite{HS19}, the payment $\pi_i$ of agent $i$ obtained in this way is equal to the maximum total weight in any simple path that originates from node $i$ in the envy graph $\EG(X,v)$.

We study the optimization problem \problem\ (standing for Subsidy Minimization for Envy-Freeness). Given an allocation instance, \problem\ aims to compute an allocation that is envy-freeable with the minimum amount of subsidies. Since the problem of computing an envy-free allocation is NP-hard, \problem\ is NP-hard as well. 

We are interested in the design of approximation algorithms for \problem. As algorithms with finite multiplicative approximation ratio are hopeless (since it is NP-hard to decide whether the minimum amount of subsidies is zero or not), we seek polynomial-time algorithms that compute an allocation that is envy-freeable with subsidies $\chi+\rho\cdot\sumv{v}$, with the approximation guarantee $\rho$ being as low as possible.

As a warmup, consider the algorithm that allocates all items to the agent $i^*$ who has maximum value for $M$ and paying a subsidy of $v_{i^*}(M)$ to every other agent $i$. Clearly, this is a polynomial-time algorithm. The allocation obtained is envy-freeable since no redistribution of the bundles (i.e., giving all items to another agent) results in higher social welfare. And the particular payments are right: agent $i^*$ is indifferent between the bundle $M$ and the payment to any other agent, while the other agents are indifferent between the (equal) payments, and prefer their payment to getting the whole bundle $M$. It can be easily verified that the algorithm guarantees an amount of at most $\chi+(n-1)\max{v}\leq \chi+(n-1)\sumv{v}$ as subsidies; this is the best guarantee of this form for this algorithm in the worst-case. 

\section{An approximation algorithm}\label{sec:alg}
We now present an algorithm that does much better. The algorithm exploits ideas that have led to polynomial-time approximation schemes for combinatorial optimization problems like {\sc Knapsack}; e.g., see~\cite{V01}. It first discretizes all valuations to multiples of a discretization parameter. In this way, the different discretized valuations an agent can have for bundles of items in the new instance is small. This allows to classify all allocations into a relatively small number of classes, each defined by specific discretized valuation levels of each agent for all bundles. Dynamic programming is used to decide the classes that are non-empty and to select a representative allocation from each class. The final allocation is selected among all representative allocations, possibly after redistributing the bundles so that social welfare (with respect to the original valuations) is maximized (in order to get envy-freeability). This requires a call to linear program (\ref{eq:LP}) to compute the minimum amount of subsidies for each representative allocation. 

The classification of allocations guarantees that the algorithm will consider a representative allocation from the class that also contains the optimal one (i.e., the allocation that is envy-freeable with the minimum amount of subsidies overall). Our analysis shows that the amount of subsidies for making the representative allocation envy-free is close to optimal. Polynomial running time for the case of a constant number of agents follows by setting the discretization parameter appropriately.

We now present our algorithm in detail. It uses an accuracy  parameter $\epsilon>0$ and initially decides the value of the discretization parameter $\discretization$ as follows:
$$\discretization=\frac{\epsilon \max{v}}{4mn^2}.$$
First, the algorithm implicitly discretizes all agent valuations by defining new valuations $\tilde{v}$ as follows: for an agent $i$ with valuation $v_i(g)$ for item $g$, the discretized valuation $\tilde{v}_i(g)$ is equal to $\left\lfloor v_i(g)/\discretization\right\rfloor \discretization$.

The algorithm uses an arbitrary ordering of the items in $M$; let $M=\{g_1, g_2, ..., g_m\}$, where the item indices are those in this ordering. The algorithm builds a table $\T$ which classifies all possible allocations of subsets of $M$. Consider an $(n^2+1)$-dimensional tuple $\tau=(t,P_{ij}, 1\leq i,j \leq n)$, where $t$ is an integer from $1$ to $m$ and $P_{ij}$ is an integer from $0$ to $\left\lfloor\max{v}/\discretization\right\rfloor$, for every pair of agents $i$ and $j$. The entry $\T(\tau)$ of the table indicates whether an allocation $A^t=(A^t_1, A^t_2, ..., A^t_n)$ of the first $t$ items $g_1, ..., g_t$ of $M$ to the $n$ agents, satisfying $\tilde{v}_i(A^t_j)=P_{ij}\discretization$ for every pair of agents $i$ and $j$, exists ($\T(\tau)=1$) or not ($\T(\tau)=0$). 

The entries of $\T$ are computed using the following recursive relation:
\begin{itemize}
	\item For a tuple $\tau=(t,P_{ij}, 1\leq i,j \leq n)$ with $t=1$, the algorithm sets $\T(\tau)=1$ if there exists $k\in [n]$ such that, for every $i\in [n]$, $\tilde{v}_i(g_1)=P_{ik}\discretization$ and $P_{ij}=0$ for every $j\not=k$. Otherwise, the algorithm sets $\T(\tau)=0$.
	\item For a tuple $\tau=(t,P_{ij}, 1\leq i,j \leq n)$ with $t>1$, the algorithm sets $\T(\tau)=1$ if there exists $k\in [n]$ and tuple $\tau'=(t-1,P'_{ij}, 1\leq i,j \leq n)$ such that, for every $i\in [n]$, $P_{ik}=P'_{ik}+\tilde{v}_i(g_t)/\discretization$ and $P_{ij}=P'_{ij}$ for every $j\not=k$. Otherwise, the algorithm sets $\T(\tau)=0$.
\end{itemize}

Essentially, each non-zero entry of $\T$ (e.g., $\T(\tau)=1$) indicates a non-empty class $\mathcal{A}_\tau$ of (possibly partial, when the first argument of $\tau$ is an integer smaller than $m$) allocations. To compute a representative complete allocation $A_\tau\in \mathcal{A}_\tau$ among those implied by the non-zero entry corresponding to the tuple $(m,P^m_{ij}, 1\leq i,j\leq n)$, the algorithm does the following for $t=m$ downto $2$. Let $k\in [n]$ be such that $\T(\tau')=1$ for a tuple $\tau'=(t-1, P^{t-1}_{ij}, 1\leq i,j \leq n)$ with $P^{t-1}_{ik}=P^t_{ik}-\tilde{v}_i(g_t)/\discretization$ and $P^{t-1}_{ij}=P^{t}_{ij}$ for every pair of agents $i$ and $j\not=k$. The algorithm assigns item $g_t$ to agent $k$ and proceeds to considering the next item. The first item $g_1$ is assigned to agent $k$ such that $\T(\tau')=1$ for a tuple $\tau'=(1, P^{1}_{ij}, 1\leq i,j \leq n)$ with $P^{1}_{ik}=\tilde{v}_i(g_1)/\discretization$ and $P^{1}_{ij}=0$ for every pair of agents $i$ and $j\not=k$.

Next, the algorithm redistributes the bundles of each allocation $A_\tau$ that represents a non-empty class $\mathcal{A}_\tau$ so that an allocation $A'_\tau$ of maximum social welfare (among those that distribute the particular bundles to the agents) is obtained (in terms of the original valuations). It solves $\LP(A'_\tau,v)$ (for the original valuations) to compute the minimum amount of subsidies that make $A'_\tau$ envy-free. Among all allocations $A'_\tau$, it outputs the one with the minimum amount of subsidies. The approximation guarantee of the algorithm is given by the next lemma.

\begin{lemma}
	Given an instance of \problem\ that has an allocation that is envy-freeable with an amount of $\chi$ as total subsidies, the algorithm computes an allocation that is envy-freeable with total subsidies of at most $\chi+4mn^2\discretization$.
\end{lemma}

\begin{proof}
Let $\tau$ be a full tuple such that $\mathcal{A}_\tau$ contains an allocation $O=(O_1, ..., O_n)$ that is envy-freeable with subsidies of $\chi$. Since $\mathcal{A}_\tau$ is non-empty, it is $\T(\tau)=1$. Let $A$ be the allocation computed by the algorithm as representative of $\mathcal{A}_\tau$ and $A'$ the allocation that is obtained after redistributing the bundles of $A$. By Theorem~\ref{thm:hs19}, $A'$ is clearly envy-freeable; we will show that the corresponding subsidies are at most $\chi+4mn^2\discretization$. Clearly, the output of the algorithm will be envy-freeable with at most this amount of subsidies.
	
	Let $\sigma\in \perm$ be the permutation over $[n]$ such that $A'_j=A_{\sigma(j)}$ for every $j\in [n]$. Let $G$ and $H$ be the envy graphs $\EG(O,v)$ and $\EG(A',v)$, respectively. 
	
	We now present the most crucial component of our analysis. It exploits the fact that both $O$ and $A$ belong to class $\mathcal{A}_\tau$ and uses the third statement of Theorem~\ref{thm:hs19}.
	
	\begin{lemma}\label{lem:edge-to-path}
		For every pair of agents $i$ and $j$, there exists a (not necessarily simple) path $p(i,j)$ from node $\sigma(i)$ to node $\sigma(j)$ such that 
		\begin{align*}
		\wgt_H(i,j) &\leq \sum_{e\in p(i,j)}{\wgt_G(e)}+4m\discretization.
		\end{align*}
	\end{lemma}
		
\begin{proof}
		In the proof, we will use the following simple claim.
	\begin{claim}\label{claim:sandwitch}
		For every agent $i$ and every two bundles $B_1$ and $B_2$ such that $\tilde{v}_i(B_1)=\tilde{v}_i(B_2)$, 
		it holds that
		\begin{align}
		-|B_2|\discretization &\leq v_i(B_1)-v_i(B_2) \leq |B_1|\discretization.
		\end{align}
	\end{claim}
	
	\begin{proof}
		First observe that, by the definition of $\tilde{v}$ and its relation to $v$, for every agent $i$ and item $g\in M$, it holds that $\tilde{v}_i(g)\leq v_i(g)\leq \tilde{v}_i(g)+\discretization$. Hence, for every bundle $B$, 
		\begin{align*}
		\tilde{v}_i(B) &\leq v_i(B)\leq \tilde{v}_i(B)+|B|\discretization. 
		\end{align*}
		The claim follows by applying this inequality for bundles $B_1$ and $B_2$ and using the fact that $\tilde{v}_i(B_1)=\tilde{v}_i(B_2)$.
	\end{proof}
	
	We use the notation $\sigma^{-1}$ to refer to the inverse permutation of $\sigma$, i.e., $\sigma^{-1}(k)=j$ when $k=\sigma(j)$. Consider the set $C$ that contains edge $(k,\sigma^{-1}(k))$ for every agent $k$ such that $k\not=\sigma^{-1}(k)$. $C$ is either empty (if $k=\sigma^{-1}(k)$ for every agent $k$) or consists of disjoint directed cycles. For an agent $i$, if $\sigma^{-1}(i)\not=i$, we denote by $C_i$ the set of nodes that are spanned by the cycle of $C$ that includes node $i$. Otherwise, we define $C_i$ to contain only node $i$.
	
	Define the (not necessarily simple) path $p(i,j)$ from node $\sigma(i)$ to node $\sigma(j)$ to contain edge $(k,\sigma(k))$ for every node $k$ in the set $C_i$ besides node $i$ and, if $i\not=\sigma(j)$, the directed edge $(i,\sigma(j))$.
	
	For every pair of agents $i$ and $j$, we have that the weight of the directed edge $(i,j)$ in $H$ is
	\begin{align*}
	\wgt_H(i,j) &\leq\wgt_H(i,j)-\sum_{k\in C_i}{\wgt_H(k,\sigma^{-1}(k))}\\ 
	&= v_i(A'_j)-v_i(A'_i)-\sum_{k\in C_i}{\left(v_k(A'_{\sigma^{-1}(k)})-v_k(A'_k)\right)}\\
	&= v_i(A_{\sigma(j)})- v_i(A_{\sigma(i)})-\sum_{k\in C_i}{\left(v_k(A_k)-v_k(A_{\sigma(k)})\right)}\\
	&\leq v_i(O_{\sigma(j)})- v_i(O_{\sigma(i)})-\sum_{k\in C_i}{\left(v_k(O_k)-v_k(O_{\sigma(k)})\right)}\\
	&\quad +\left(|A_{\sigma(j)}|+|O_{\sigma(i)}|+\sum_{k\in C_i}{|O_k|}+\sum_{k\in C_i}{|A_{\sigma(k)}|}\right)\discretization\\
	&\leq v_i(O_{\sigma(j)})- v_i(O_{\sigma(i)})-\sum_{k\in C_i}{\left(v_k(O_k)-v_k(O_{\sigma(k)})\right)}+4m\discretization\\
	&=v_i(O_{\sigma(j)})- v_i(O_i)+\sum_{k\in C_i\setminus \{i\}}{\left(v_k(O_{\sigma(k)})-v_k(O_k)\right)}+4m\discretization\\
	&=\wgt_G(i,\sigma(j))+\sum_{k\in C_i\setminus\{i\}}{\wgt_G(k,\sigma(k))}+4m\discretization\\
	&=\sum_{e\in p(i,j)}{\wgt_G(e)}+4m\discretization.
	\end{align*}
	The first inequality follows since $C_i$ consists of node $i$ only (when $i=\sigma(i)$) or the edges $(k,\sigma^{-1}(k))$ for $k\in C_i$ form a directed cycle of non-positive total weight in $H$. The second inequality follows by applying Claim~\ref{claim:sandwitch} (recall that both allocations $A$ and $O$ belong to the class $\mathcal{A}_\tau$ and, hence, $\tilde{v}_\ell(A_q)=\tilde{v}_\ell(O_q)$ for every pair of agents $\ell$ and $q$). The third inequality follows since the bundles $A_{\sigma(k)}$ (respectively, $O_k$) for $k\in C_i$ are disjoint. The equalities are obvious or follow by the definition of the weights.
\end{proof}

	Now, let $\pi'$ and $\pi$ be the solutions of $\LP(A',v)$ and $\LP(O,v)$, respectively. Hence, $\chi=\SUB(O,v)=\sum_{i=1}^n{\pi_i}$. We will use Lemma~\ref{lem:edge-to-path} to argue that 
	\begin{align}\label{eq:close}
	\pi'_i\leq \pi_{\sigma(i)}+4mn\discretization.
	\end{align}
	This will yield $$\SUB(A',v)=\sum_{i=1}^{n}{\pi'_i}\leq \sum_{i=1}^{n}{\left(\pi_{\sigma(i)}+4mn\discretization\right)} = \chi+4mn^2\discretization,$$
	completing the proof.
	
Recall from Theorem~\ref{thm:hs19} that the payment $\pi'_\ell$ (respectively, $\pi_\ell$) is equal to the maximum path weight over all simple paths that originate from node $\ell$ in graph $H$ (respectively, graph $G$). Let $Q_\ell$ be the corresponding simple path that is destined for some node $s$ (and originates from node $\ell$), i.e., $\pi'_\ell=\sum_{e\in Q_\ell}{\wgt_H(e)}$. We construct the (not necessarily simple) path $P_\ell$ from node $\sigma(\ell)$ to node $\sigma(s)$ of $G$ that consists of path $p(i,j)$ for every directed edge $(i,j)$ in the path $Q_\ell$. Using Lemma~\ref{lem:edge-to-path}, we get
	\begin{align}\nonumber
	\pi'_\ell &=\sum_{e\in Q_\ell}{\wgt_H(e)} \leq \sum_{e\in Q_\ell}{\left(\sum_{e'\in p(e)}{\wgt_G(e')}+4m\discretization\right)}\\\label{eq:bound-by-P_t}
	& \leq \sum_{e\in Q_\ell}{\sum_{e'\in p(e)}{\wgt_G(e')}}+4mn\discretization =\sum_{e\in P_\ell}{\wgt_G(e)}+4mn\discretization. 
	\end{align}
	The second inequality follows since path $Q_\ell$ is simple (and, hence, contains at most $n-1$ edges). Now, create the simple path $P'_\ell$ from node $\sigma(\ell)$ to node $\sigma(s)$ by removing the cycles in $P_\ell$. Since graph $G$ does not have any directed cycles of positive total weight (by Theorem~\ref{thm:hs19}), we have $\wgt_G(P_\ell)\leq \wgt_G(P'_\ell)$. Now, (\ref{eq:bound-by-P_t}) yields
	\begin{align*}
	\pi'_\ell & \leq \sum_{e\in P'_\ell}{\wgt_G(e)}+4mn\discretization,
	\end{align*}
	which implies (\ref{eq:close}) since $P'_\ell$ is a simple path that originates from node $\sigma(i)$.
\end{proof}

The running time of the algorithm depends on the number of table entries, the number of steps required for computing each table entry using the recursive relation, the number of steps required to compute a representative allocation for a non-empty allocation class, the redistribution time, and the time required to solve the linear programs.

The dimensions of the table $\T$ are $m$ for the first one that enumerates over all items, and at most $1+ \left\lfloor\max{v}/\discretization\right\rfloor=1+\frac{4mn^2}{\epsilon}$ for each of the other dimensions. Overall, the size of the table is $\bigO\left(\left(\frac{m}{\epsilon}\right)^{n^2+1}\right)$. The computation of each table entry using the recursive relation needs the values in $n^2$ table entries that have previously computed. In a representative allocation, the agent in which each of the $m$ items is allocated requires time $n^2$ as well, i.e., time $\bigO(m)$ in total. The redistribution of the bundles can be implemented using a matching computation in a complete edge-weighted bipartite graph that has a node for each agent and for each bundle and the weight of an edge indicates the valuation of an agent for a bundle. As $n$ is constant, this takes constant time. Also, the linear programs have constant size. In general, since $n$ is a constant, it is ignored in the $\bigO$ notation unless it appears in the exponent. The above discussion is summarized in the next statement.

\begin{theorem}\label{thm:upper}
	Let $\epsilon>0$ be the accuracy parameter used by the algorithm. Given an instance of \problem\ consisting of a constant number $n$ of agents with valuations $v$ over a set $M$ of $m$ items that has an envy-freeable allocation using an amount $\chi$ of subsidies, the algorithm runs in time $\bigO\left(\left(m/\epsilon\right)^{n^2+2}\right)$ and computes an allocation that is envy-freeable using a total subsidy of at most $\chi+\epsilon \max{v}$.
\end{theorem}

\section{Hardness of approximating \problem}\label{sec:hardness}
In this section, we show that approximation guarantees like the one in the statement of Theorem~\ref{thm:upper} are not possible when the number of agents is part of the input.

\begin{theorem}\label{thm:hardness}
	Approximating \problem\ within an additive term of $3\cdot 10^{-4} \sumv{v}$ is NP-hard.
\end{theorem}

We prove Theorem~\ref{thm:hardness} by presenting a reduction from Maximum 3-Dimensional Matching (MAX-3DM). An instance of MAX-3DM consists of three disjoint sets of elements $A=\{a_1, a_2, ..., a_n\}$, $B=\{b_1, b_2, ..., b_n\}$, and $C=\{c_1, c_2, ..., c_n\}$, each of size $n$, and a set $T$ of $m$ triplets of the form $(a_i,b_j,c_k)$ with $a_i\in A$, $b_j\in B$, and $c_k\in C$. The objective is to compute a disjoint subset of $T$ (or, simply, a 3D matching) of maximum size. The problem is well-known to be NP-hard not only to solve exactly \citep{GJ79} but also to approximate~\citep{K91}.

We will use the inapproximability result of~\cite{CC06}, which applies to bounded instances of MAX-3DM in which each element appears in exactly two triplets (i.e., $m=2n$); we will refer to this restriction of MAX-3DM as MAX-3DM-2. In particular, \cite{CC06} show that it is NP-hard to distinguish between instances of MAX-3DM-2 with a 3D matching of size at least $K$ and instances of MAX-3DM-2 in which any 3D matching has size at most $K-0.01n$.\footnote{This statement is actually weaker than the one proved by~\cite{CC06}. However, it suffices for our purpose to prove hardness of approximation. Note that we have made no particular attempt to optimize our inapproximabity threshold.}

\subsection{The reduction}
We present our reduction and full proof for the case $\chi>0$. We omit the case $\chi=0$, which requires a minor modification of the reduction. 
%At the end of the section, we discuss the minor modifications required for the case $\chi=0$.
On input an instance of MAX-3DM-2, our reduction constructs in polynomial time an instance of \problem, in which the minimum amount of susbsidies that can make some allocation envy-free is exactly $\chi(1+\max\{K-L,0\})$, where $L$ is the size of the maximum 3D matching in the MAX-3DM-2 instance. Using the result of~\cite{CC06}, we will get that it is NP-hard to distinguish between \problem\ instances in which the minimum amount of subsidies is at most $\chi$ and instances in which it is at least $\chi(1+0.01n)$. Hence, \problem\ will be proved to be NP-hard to approximate within $0.01n\chi$. Our construction will be such that $\sumv{v}< 30n\chi$. In this way, we will obtain a hardness of approximating \problem\ within an additive term of (at least) $3\cdot 10^{-4} \sumv{v}$, as desired.

Our reduction is as follows. Given an instance of MAX-3DM-2 consisting of sets of elements $A$, $B$, and $C$, each of size $n$, and a set of $2n$ triplets $T$, the instance of \problem\ has 
\begin{itemize}
	\item three agents 1, 2, and 3,
	\item three agents $J_1(t)$, $J_2(t)$, and $J_3(t)$ for every triplet $t\in T$,
	\item an item $A_i$ for every element $a_i\in A$,
	\item an item $B_i$ for every element $b_i\in B$,
	\item an item $\Gamma_i$ for every element $c_i\in C$,
	\item three items $\Delta_t$, $Z_t$, and $\Theta_t$ for every triplet $t\in T$, and 
	\item an additional item $\Lambda$.
\end{itemize}

The agents $J_1(t)$, $J_2(t)$, and $J_3(t)$ that correspond to the triplet $t=(a_i,b_j,c_j)$ have valuations $0$ for all items besides the items $A_i$, $B_j$, $\Gamma_k$, $\Delta_t$, $Z_t$, and $\Theta_t$. Agents 1, 2 have valuation $0$ for all items besides item $\Lambda$ and agent 3 has valuation zero for all items besides item $\Lambda$ and items $\Theta_t$ for $t\in T$. Their remaining valuations are as follows:

\begin{center}
	\begin{tabular}{r|ccccccc}
		& $A_i$ & $B_j$ & $\Gamma_k$ & $\Delta_t$ & $Z_t$ & $\Theta_t$ & $\Lambda$ \\\hline
		$1$       & $0$ & $0$ & $0$ & $0$ & $0$ & $0$ & $\chi$ \\
		$2$      &  $0$ & $0$ & $0$ & $0$ & $0$ & $0$ & $\chi K$ \\
		$3$      &  $0$ & $0$ & $0$ & $0$ & $0$ & $\chi$ & $\chi K$ \\
		$J_1(t)$ & $\chi$ & $\chi$ & $\chi$ & $3\chi$ & $3\chi$ & $0$ & $0$\\
		$J_2(t)$ & $0$ & $0$ & $0$ & $\chi$ & $\chi$ & $\chi$ & $0$\\
		$J_3(t)$ & $0$ & $0$ & $0$ & $0$ & $\chi$ & $0$ & $0$
	\end{tabular}
\end{center}
Recall that each element belongs to exactly two triplets. Hence, two agents have positive value for item $A_i$ (similarly for items $B_j$ and $\Gamma_k$): agents $J_2(t_1)$ and $J_2(t_2)$ such that the triplets $t_1$ and $t_2$ contains element $a_i$ (similarly for elements $b_j$ and $c_k$). It is easy to see that either two or three agents have positive value for each item. For every triplet $t$, the agents $J_1(t)$, $J_2(t)$, and $J_3(t)$ have total valuation $9\chi$, $3\chi$, and $\chi$, respectively. Taking into account that $K\leq n$, we obtain that $\sumv{v}<30n\chi$.

\subsection{Lower bound on subsidies}
Consider an instance of \problem\ constructed by our reduction and let $X$ be an envy-freeable allocation in it. We will first lower-bound the minimum amount of subsidies that make $X$ envy-free. First observe that $X$ cannot give item $\Lambda$ to agent 1; in that case, exchanging the bundles of agents 1 and 2 would result to an increase of the social welfare and, hence, $X$ would not be envy-freeable. If $X$ gives item $\Lambda$ to agent 3, agents 1 and 2 would need subsidies of at least $\chi$ and $\chi K$, respectively, so that they do not envy agent 3. Hence, $\SUB(A,v)\geq \chi(1+K)$ in this case. 

In the following, we will lower-bound the minimum total subsidies that make $X$ envy-free assuming that item $\Lambda$ is given to agent 2. Let $\theta$ be the number of items $\Theta_t$ for $t\in [2n]$ agent 3 gets. Then, agent 3 should be given a subsidy of at least $\chi\max\{K-\theta,0\}$ so that she does not envy agent 2. Agent 1 needs a subsidy of $\chi\max\{K-\theta,1\}$ so that she does not envy agents 1 and 2.

For a triplet $t=(a_i,b_j,c_k)$ in the original instance of MAX-3DM-2, we call it {\em full} if all items $A_i$, $B_j$, and $\Gamma_k$ (which correspond to the elements of the triplet) have been allocated to  the agents $J_1(t)$, $J_2(t)$, or $J_3(t)$. Otherwise, we call it {\em partial}. We call $t$ {\em supported} if item $\Theta_t$ has been allocated to agent $J_2(t)$; otherwise, we call $t$ {\em unsupported}.

In the next four claims, we lower-bound the total amount of subsidies the agents $J_1(t)$, $J_2(t)$, and $J_3(t)$ of a triplet $t$ need, depending of the type of $t$. 

\begin{claim}\label{claim:1}
	The agents $J_1(t)$, $J_2(t)$, and $J_3(t)$ of a full and supported triplet $t$ need subsidies of at least $\chi\max\{K-\theta-2,0\}$.
\end{claim}

\begin{proof}
	Consider a full and supported triplet $t$. If agent $J_2(t)$ has value at most $2\chi$ (i.e., getting $\Theta_t$ and at most one of the items $\Delta_t$ and $Z_t$), then she needs a subsidy of at least $\chi\max\{K-\theta-2,0\}$ so that she does not envy agent 3. If agent $J_2(t)$ has value $3\chi$ by getting both items $\Delta_t$ and $Z_t$ in addition to $\Theta_t$, she needs a subsidy of at least $\chi\max\{K-\theta-3,0\}$, while then agents $J_1(t)$ and $J_3(t)$ need subsidies of at least $3\chi+\chi\max\{K-\theta-3,0\}$ and $\chi+\chi\max\{K-\theta-3,0\}$, respectively, so that they do not envy agent $J_2(t)$. In both cases, the total amount of subsidies of the agents $J_1(t)$, $J_2(t)$, and $J_3(t)$ is at least $\chi\max\{K-\theta-2,0\}$.
\end{proof}

\begin{claim}\label{claim:2}
	The agents $J_1(t)$, $J_2(t)$, and $J_3(t)$ of a full and unsupported triplet $t$ need subsidies of at least $\chi\max\{K-\theta-1,0\}$.
\end{claim}

\begin{proof}
	Consider a full and unsupported triplet $t$. If agent $J_2(t)$ has value at most $\chi$ (i.e., getting at most one of the items $\Delta_t$ and $Z_t$), then she needs a subsidy of at least $\chi\max\{K-\theta-1,0\}$ so that she does not envy agent 3. If agent $J_2(t)$ has value $2\chi$ by getting both items $\Delta_t$ and $Z_t$, she needs a subsidy of at least $\chi\max\{K-\theta-2,0\}$, while then agents $J_1(t)$ and $J_3(t)$ need subsidies of at least $3\chi+\chi\max\{K-\theta-2,0\}$ and $\chi+\chi\max\{K-\theta-2,0\}$, respectively, so that they do not envy agent $J_2(t)$. In both cases, the total amount of subsidies of agents $J_1(t)$, $J_2(t)$, and $J_3(t)$ is at least $\chi\max\{K-\theta-1,0\}$. 
\end{proof}

\begin{claim}\label{claim:3}
	The agents $J_1(t)$, $J_2(t)$, and $J_3(t)$ of a partial and supported triplet $t$ need subsidies of at least $\chi\max\{K-\theta-1,0\}$.
\end{claim}

\begin{proof}
	Let $t$ be a partial and supported triplet. If agent $J_2(t)$ does not get items $\Delta_t$ and $Z_t$, then she gets only a value of $\chi$ from item $\Theta_t$ and needs a subsidy of at least $\chi\max\{K-\theta-1,0\}$ so that she does not envy agent 3. 
	
	If agent $J_2(t)$ gets item $\Delta_t$ but not item $Z_t$, she needs a subsidy of $\chi\max\{K-\theta-2,0\}$ so that she does not envy agent 3. Then, if agent $J_1(t)$ does not get item $Z_t$, her value is at most $2\chi$ (from at most two of the items $A_i$, $B_j$, and $\Gamma_k$) and needs a subsidy of $\chi+\chi\max\{K-\theta-2,0\}$ so that she does not envy agent $J_2(t)$. If agent $J_3(t)$ does not get item $Z_t$, she needs a subsidy of at least $\chi+\chi\max\{K-\theta-2,0\}$ so that she does not envy agent $J_2(t)$. 
	
	If agent $J_2(t)$ gets item $Z_t$ but not $\Delta_t$, she needs a subsidy of $\chi\max\{K-\theta-2,0\}$ so that she does not envy agent 3 and agent $J_3(t)$ needs a subsidy of at least $\chi+\chi\max\{K-\theta-2,0\}$ so that she does not envy agent $J_2(t)$. 
	
	Finally, if agent $J_2(t)$ gets items $\Delta_t$ and $Z_t$, her value is $3\chi$ and needs a subsidy of at least $\chi\max\{K-\theta-3,0\}$ so that she does not envy agent 3. Then, each of agents $J_1(t)$ and $J_3(t)$ need a subsidy of at least $\chi+\chi\max\{K-\theta-3,0\}$ so that they do not envy agent $J_2(t)$.
	
	In all cases, the total amount of subsidies the agents $J_1(t)$, $J_2(t)$, and $J_3(t)$ need is at least $\chi\max\{K-\theta-1,0\}$. 
\end{proof}

\begin{claim}\label{claim:4}
	The agents $J_1(t)$, $J_2(t)$, and $J_3(t)$ of a partial and unsupported triplet $t$ need subsidies of at least $\chi\max\{K-\theta,1\}$.
\end{claim}

\begin{proof}
	Let $t$ be a partial and unsupported triplet. If agent $J_2(t)$ gets both items $\Delta_t$ and $Z_t$, she needs a subsidy of $\chi\max\{K-\theta-2,0\}$ so that she does not envy agent 3, while agents $J_1(t)$ and $J_3(t)$ would then need subsidies of at least $4\chi+\chi\max\{K-\theta-2,0\}$ and $\chi+\chi\max\{K-\theta-2,0\}$, respectively, so that they do not envy agent $J_2(t)$. 
	
	If agent $J_2(t)$ gets only item $\Delta_t$, she needs a subsidy of $\chi\max\{K-\theta-1,0\}$ so that she does not envy agent 3. Then, the agent who does not get item $Z_t$ among $J_1(t)$ and $J_3(t)$ would need a subsidy of at least $\chi+\chi\max\{K-\theta-1,0\}$ so that she does not envy agent $J_2(t)$. 
	
	If agent $J_2(t)$ gets only item $Z_t$, she needs a subsidy of $\chi\max\{K-\theta-1,0\}$ so that she does not envy agent 3, while agent $J_3(t)$ needs a subsidy of at least $\chi+\chi\max\{K-\theta-1,0\}$ so that she does not envy agent $J_2(t)$. 
	
	Finally, if agent $J_2(t)$ gets no item (among $\Delta_t$ and $Z_t$), she needs a subsidy of at least $\chi$ so that she does not envy the agents who get items $\Delta_t$ and $Z_t$ and a subsidy of at least $\chi\max\{K-\theta,0\}$ so that she does not envy agent 3. 
	
	In all cases, the total amount of subsidies the agents $J_1(t)$, $J_2(t)$, and $J_3(t)$ need is at least $\chi\max\{K-\theta,1\}$.
\end{proof}

We now denote by $L_1$, $L_2$, $P_1$, and $P_2$, the number of full and supported, full and unsupported, partial and supported, and partial and unsupported triplets defined by $X$, respectively. Notice that the full triplets form a 3D matching. Denoting by $L$ the maximum size over all 3D matchings  of the MAX-3DM-2 instance, we have $L\geq L_1+L_2$. Using Claims~\ref{claim:1}-\ref{claim:4}, and our observations for agents 1 and 3, we have that the total amount of subsidies $X$ needs to become envy-free is 
\begin{align}\nonumber
 \SUB(X,v)&\geq \chi \left(L_1\max\{K-\theta-2,0\}+L_2\max\{K-\theta-1,0\}\right.\\\nonumber
& \quad \left.+P_1\max\{K-\theta-1,0\}+P_2\max\{K-\theta,1\}\right.\\\label{eq:sub}
& \quad \left.+\max\{K-\theta,0\}+\max\{K-\theta,1\}\right). 
\end{align}

We will distinguish between two cases for $K-\theta$. If $K-\theta\geq 2$, (\ref{eq:sub}) yields
\begin{align*}
\SUB(X,v) &\geq \chi \left(L_2+P_1+2P_2+4\right)=\chi \left(2n-L_1+P_2+4\right) \geq \chi(1+\max\{K-L,0\}).
\end{align*}
Now, notice that $\theta$, the number of items $\Theta_t$ agent 3 gets in $X$ is upper-bounded by the number of unsupported triplets, i.e., $\theta \leq L_2+P_2$. Thus, if $K-\theta\leq1$, (\ref{eq:sub}) yields
 \begin{align*}
 \SUB(X,v) &\geq \chi \left(P_2+K-\theta+1\right) \geq \chi \left(K-L_2+1\right) \geq \chi(1+\max\{K-L,0\}).
 \end{align*}

We conclude that the minimum amount of subsidies necessary to make $X$ envy-free is at least $\chi(1+\max\{K-L,0\})$.

\subsection{Upper bound on minimum subsidies}
We now present our upper bound on the minimum amount of subsidies for envy-freeness. Given a 3D matching $\mathcal{M}$ of maximum size $L$ in the MAX-3DM-2 instance, we will construct an allocation for the \problem\ instance and will show that it is envy-freeable with an amount of subsidies equal to $\chi(1+\max\{K-L,0\})$.
	
	For defining the allocation, we partition $T\setminus \mathcal{M}$ in two disjoint sets of triplets $T_1$ and $T_2$ of size $2n-\max\{K,L\}$ and $\max\{K-L,0\}$, respectively.
\begin{itemize}
	\item For every triplet $t=(a_i,b_j,c_k) \in \mathcal{M}$, agent $J_1(t)$ gets items $A_i$, $B_j$, and $\Gamma_k$, agent $J_2(t)$ gets item $\Delta_t$ and agent $J_3(t)$ gets item $Z_t$.
	\item For every triplet $t=(a_i,b_j,c_k) \not\in \mathcal{M}$, let $F(t)$ be the set of items that correspond to the elements of $t$ that have not been included in triplets of $\mathcal{M}$. Note that, due to the maximality of $\mathcal{M}$, $F(t)$ has zero, one, or two elements among $A_i$, $B_j$, and $\Gamma_k$. For every triplet $t=(a_i,b_j,c_k) \in T_1$, agent $J_1(t)$ gets item $\Delta_t$, agent $J_2(t)$ gets the items in $F(t)$, if any, and item $\Theta_t$, and agent $J_3(t)$ gets item $Z_t$.
	\item 
	For every triplet $t=(a_i,b_j,c_k) \in T_2$, agent $J_1(t)$ gets item $\Delta_t$, agent $J_2(t)$ gets the items in $F(t)$, if any, and agent $J_3(t)$ gets item $Z_t$.
	\item Agent 3 gets item $\Theta_t$ for every triplet $t\in \mathcal{M}\cup T_2$.
	\item Agent 2 gets item $\Lambda$.
	\item Agent 1 gets no items.
\end{itemize}
We claim that the allocation above is envy-freeable by assigning a subsidy of $\chi$ to agent 1 and a subsidy of $\chi$ to agent $J_2(t)$ for every triplet $t\in T_2$ (if any). 

Indeed, agent 1 has positive value only for item $\Lambda$, which is given to agent 2, who gets no subsidy. Also, no other agent gets a subsidy more than the subsidy $\chi$ that is given to agent 1. Hence, agent 1 is not envious. Agent 2 gets item $\Lambda$, which is the only item she values positively and much higher than the subsidy given to any other agent. Hence, agent 2 is not envious either. Agent 3 gets exactly $\max\{K,L\}$ items of total value $\chi\max\{K,L\}$. She does not envy agent 2 who gets item $\Lambda$ (which agent 3 values for $\chi K$) since no subsidy is given to agent 2. Clearly, the value of agent 3 is much higher than the subsidy given to any other agent.

Consider a triplet $t=(a_i,b_j,c_k) \in \mathcal{M}$. Agent $J_1(t)$ has a value of $3\chi$ for the items $A_i$, $B_j$, and $\Gamma_k$ she gets. The remaining items for which she has positive valuation of $3\chi$ have been given to agents $J_2(t)$ and $J_3(t)$, respectively. Since these agents do not get subsidies, agent $J_1(t)$ is not envious of them. Clearly, agent $J_1(t)$ is not envious of any other agent since she has zero value for all other items and no agent gets a subsidy more than $\chi$. Agent $J_3(t)$ gets item $Z_t$, the only item for which she has positive value and does not envy any other agent since no one gets a subsidy higher than $\chi$. Agent $J_2(t)$ gets a value of $\chi$ from item $\Delta_t$ and does not envy agent $J_3(t)$, who gets item $Z_t$, or agent 3, who gets item $\Theta_t$, as these agents receive no subsidy. Clearly, agent $J_2(t)$ envies no other agent.

Now consider a triplet $t=(a_i,b_j,c_k) \not \in \mathcal{M}$. Agent $J_1(t)$ has a value of $3\chi$ for the item $\Delta_t$ she gets. The remaining items for which she has positive valuation have been allocated as follows. Item $Z_t$ has been given to agent $J_3(t)$; clearly, agent $J_1(t)$ is not envious of $J_3(t)$ since the latter gets no subsidies. The items in $F(t)$ have been given to agent $J_2(t)$. Again, agent $J_1(t)$ is not envious of $J_2(t)$ since $F(t)$ contains at most two items (which agent $J_1(t)$ values for $\chi$ each) and agent $J_2(t)$ gets a subsidy of zero (if $t\in T_1$) or $\chi$ (if $i\in T_2$). Clearly, $J_1(t)$ does not envy any other agent. Agent $J_3(t)$ gets item $Z_t$, the only item for which she has positive value and does not envy any other agent since no one gets a subsidy higher than $\chi$. Agent $J_2(t)$ gets a value of $\chi$, either from item $\Theta_t$ (if $t\in T_1$) or as subsidy (if $t\in T_2$), and does not envy agent $J_1(t)$ who gets item $\Delta_t$ or agent 3 who gets item $\Theta_t$ only when $t\in T_2$; recall that these two agents never get subsidies. Again, agent $J_2(t)$ envies no other agent.

\subsection{Adapting the proof for the case $\chi=0$}
The modification required in our reduction so that it covers the case $\chi=0$ as well is to remove agent 1 and replace $\chi$ with 1 in the definition of valuations. In particular, the agents $J_1(t)$, $J_2(t)$, and $J_3(t)$ that correspond to the triplet $t=(a_i,b_j,c_j)$ have valuations $0$ for all items besides the items $A_i$, $B_j$, $\Gamma_k$, $\Delta_t$, $Z_t$, and $\Theta_t$. Agent 2 has valuation $0$ for all items besides item $\Lambda$ and agent 3 has valuation zero for all items besides item $\Lambda$ and items $\Theta_t$ for $t\in T$. The remaining valuations are now as follows:

\begin{center}
	\begin{tabular}{r|ccccccc}
		& $A_i$ & $B_j$ & $\Gamma_k$ & $\Delta_t$ & $Z_t$ & $\Theta_t$ & $\Lambda$ \\\hline
		$2$      &  $0$ & $0$ & $0$ & $0$ & $0$ & $0$ & $K$ \\
		$3$      &  $0$ & $0$ & $0$ & $0$ & $0$ & $1$ & $ K$ \\
		$J_1(t)$ & $1$ & $1$ & $1$ & $3$ & $3$ & $0$ & $0$\\
		$J_2(t)$ & $0$ & $0$ & $0$ & $1$ & $1$ & $1$ & $0$\\
		$J_3(t)$ & $0$ & $0$ & $0$ & $0$ & $1$ & $0$ & $0$
	\end{tabular}
\end{center}
The same reasoning as in our proof for the case $\chi\not=0$ gives a minimum amount of subsidies for the \problem\ instance of exactly $\max\{K-L,0\}$, where $L$ is the maximum 3D matching size in the MAX-3DM-2 instance. In this way, we get that \problem\ is NP-hard to approximate within $0.01n$ (i.e., it is NP-hard to distinguish between envy-free instances and instances that need subsidies of $0.01n$) and the construction satisfies $\sumv{v}<30n$. This yields the desired inapproximability result in the statement of Theorem~\ref{thm:hardness} for the case $\chi=0$ as well.

\section{Concluding remarks}\label{sec:open}
We have initiated the study of the optimization problem \problem. The challenging open problem that deserves investigation is to close the gap between the trivial approximation guarantee of $n-1$ in Section~\ref{sec:prelim} and our negative result for super-constant numbers of agents in Section~\ref{sec:hardness}. Unfortunately, more sophisticated existing algorithms, such as the recent one by \cite{BDN+19}, do not lead to better approximations. 

We remark that $\max{v}$ could be used alternatively to $\sumv$ in the definition of the approximation guarantees of \problem. Actually, the guarantee for our dynamic programming algorithm is stated in terms of $\max{v}$. We can express the rest of our results using $\max{v}$ as well. First, the trivial algorithm presented at the end of Section~\ref{sec:prelim} uses an amount of $\chi + (n-1)\max{v}$ as subsidies. Second, an adaptation of the current proof of the inapproximability result can easily give that approximating \problem\ within an additive term of $c \times \max{v}$ for a constant $c$ is NP-hard. The important observation is that $\max{v}<n \chi$ (or $\max{v}<n$ when $\chi=0$) in our construction. Then, distinguishing between \problem\ instances in which the minimum amount is at most $\chi$ and at least $\chi (1+0.01n)$ (or at least $0.01n$ when $\chi=0$) requires to distinguish between \problem\ instances in which the minimum amount is at most $\chi$ and at least $\chi + 0.01\max{v}$. So, the inapproximability constant is a bit higher in this case. The main advantage of adopting $\sumv$ is that it makes the problem of computing the tight approximation factor more challenging.

Interestingly, an advantage of the trivial algorithm is that the particular payments incentivize the agents to report their valuations truthfully. What is the best possible approximation guarantee that can be obtained for \problem\ by truthful algorithms? Unfortunately, a simple application of Myerson's characterization in single-item settings~\citep{M81} indicates that no approximation guarantee better than $n-1$ is possible. Indeed, consider instances with a single item. By the characterization of envy-freeable allocations by~\cite{HS19} (i.e., the second statement in Theorem~\ref{thm:hs19}), we know that the agent with the highest valuation should get the item. Then, Myerson's characterization for truthful mechanisms in single parameters environments and our requirement for non-negative payments give us the specific form payments should have so that truthful reporting is a dominant strategy for all agents when this algorithm is used: if the agent $i$ who gets the item receives payment of $p\geq 0$, agent $t$ should get a payment of exactly $p+v_i-v_t$, where $v_i$ and $v_t$ are the payments of agents $i$ and $t$. Now, consider specifically the instance in which one agent has value $1$ for the item, and all other agents have value $0$. Truthfulness requires (at least) a unit of subsidy to each agent that does not get the item (i.e., total subsidies of $n-1$ while $\sumv=1$), even though there is clearly an allocation that is envy-free without any payments. This yields the claimed lower bound of $n-1$ in the approximation guarantee.

\paragraph{Acknowledgements.} Part of this work was done while the authors were at the Department of Computer Engineering and Informatics, University of Patras, Patras, Greece.

%% The file named.bst is a bibliography style file for BibTeX 0.99c
\bibliography{neg}

\end{document}